\documentclass[aps,10pt,twocolumn,showpacs,pra,citeautoscript,amsmath,amssymb,floatfix,superscriptaddress]{revtex4-1}

\usepackage{amsmath}
\usepackage{amssymb}
\usepackage{epsfig}
\usepackage{color,epstopdf}
\usepackage{amsthm, thm-restate}
\usepackage{hyperref}
\usepackage{csquotes}
\usepackage{mathtools}
\usepackage{xcolor}
\usepackage{comment}
\usepackage{float}
	
\usepackage[normalem]{ulem}
\usepackage[caption=false]{subfig}
\usepackage{algpseudocode}

\newtheorem{theorem}{Theorem}
\newtheorem*{theorem*}{Theorem}

\newtheorem{lemma}{Lemma}
\newtheorem*{lemma*}{Lemma}

\newtheorem{definition}{Definition}

\def\dim{{\mathrm{dim}}}
\def\phys{{\mathrm{phys}}}
\def\sym{{\mathrm{sym}}}

\def\Tr{{\mathrm{Tr}}}
\def\tr{{\mathrm{tr}}}
\newcommand{\ch}{\mathcal H}
\newcommand{\cu}{\mathcal U}
\newcommand{\cg}{\mathcal G}

\newcommand{\ket}[1]{|#1\rangle}
\newcommand{\bra}[1]{\langle #1|}
\newcommand{\bracket}[2]{\langle #1|#2\rangle}
\newcommand{\ketbra}[2]{|#1\rangle\langle #2|}

\begin{document}
\title{Entanglement-asymmetry correspondence for internal quantum reference frames}

\author{Anne-Catherine de la Hamette}
\email{AnneCatherine.delaHamette@univie.ac.at}
\affiliation{Institute for Quantum Optics and Quantum Information,
Austrian Academy of Sciences, Boltzmanngasse 3, A-1090 Vienna, Austria}
\affiliation{Vienna Center for Quantum Science and Technology (VCQ), Faculty of Physics, University of Vienna, Vienna, Austria}
\author{Stefan L.\ Ludescher}
\email{Stefan.Ludescher@oeaw.ac.at}
\affiliation{Institute for Quantum Optics and Quantum Information,
Austrian Academy of Sciences, Boltzmanngasse 3, A-1090 Vienna, Austria}
\affiliation{Vienna Center for Quantum Science and Technology (VCQ), Faculty of Physics, University of Vienna, Vienna, Austria}
\author{Markus P.\ M\"uller}
\email{Markus.Mueller@oeaw.ac.at}
\affiliation{Institute for Quantum Optics and Quantum Information,
Austrian Academy of Sciences, Boltzmanngasse 3, A-1090 Vienna, Austria}
\affiliation{Vienna Center for Quantum Science and Technology (VCQ), Faculty of Physics, University of Vienna, Vienna, Austria}
\affiliation{Perimeter Institute for Theoretical Physics, 31 Caroline Street North, Waterloo, Ontario N2L 2Y5, Canada}

\date{December 22, 2022}

\begin{abstract}
In the quantization of gauge theories and quantum gravity, it is crucial to treat reference frames such as rods or clocks not as idealized external classical relata, but as internal quantum subsystems. In the Page-Wootters formalism, for example, evolution of a quantum system $S$ is described by a stationary joint state of $S$ and a quantum clock, where time-dependence of $S$ arises from conditioning on the value of the clock. Here, we consider (possibly imperfect) internal quantum reference frames $R$ for arbitrary compact symmetry groups, and show that there is an exact quantitative correspondence between the amount of entanglement in the invariant state on $RS$ and the amount of asymmetry in the corresponding conditional state on $S$. Surprisingly, this duality holds exactly regardless of the choice of coherent state system used to condition on the reference frame. Averaging asymmetry over all conditional states, we obtain a simple representation-theoretic expression that admits the study of the quality of imperfect quantum reference frames, quantum speed limits for imperfect clocks, and typicality of asymmetry in a unified way. Our results shed light on the role of entanglement for establishing asymmetry in a fully symmetric quantum world.
\end{abstract}

\maketitle

\emph{Introduction.} In a quantum world with fundamental symmetries, all physical quantities must ultimately be understood as relative to some frame of reference which is itself a quantum system. This simple insight has long been regarded as relevant~\cite{AharonovSusskind,DeWitt,AharonovKaufherr,Wigner,ArakiYanase} for the quantization of gravity and other gauge theories~\cite{RovelliQG,Thiemann,HenneauxTeitelboim,RovelliRef,RovelliObservable,Isham,Dittrich,Marolf}, for the study of asymmetry in quantum information theory~\cite{BRWS,Marvian,MarvianSpekkens,GourSpekkens,GourMarvianSpekkens,PalmerGirelliBartlett,LoveridgeMiyaderaBusch,Loveridge,SmithPianiMann,Smith,Angelo}, and for quantum thermodynamics~\cite{LostaglioJenningsRudolph,LostaglioKorzekwa,Cwiklinski,Erker,GourJenningsBuscemi}. Recently, it has led to a surge of interest in the quantum foundations community on the behavior of quantum systems under transformations between such ``internal'' quantum reference frames~\cite{Giacomini,HametteGalley,GiacominiSpacetime,Streiter,CastroRuiz,Mikusch,GiacominiSpin,Ballesteros,KrummHoehnMueller,HoehnKrummMueller,Hardy18,ZychBrukner,Hardy19,GBEEP,GB2109,HoehnSubsystem,CastroRuizOreshkov,DGHLM21,HoehnSubsystem,CastroRuizOreshkov,DGHLM21,Vanrietvelde,Vanrietvelde2,HoehnUniverse,Vanrietvelde3,Trinity,HSLRel,DGHLM21,HoehnKrummMueller}. Among other results, this has led to proposals for quantum formulations of Einstein's equivalence principle~\cite{Hardy18,ZychBrukner,Hardy19,GBEEP,GB2109} and to insights into the relativity of the notion of subsystem~\cite{HoehnSubsystem,CastroRuizOreshkov,DGHLM21}.

It has been shown that several of these approaches can be unified and generalized in a ``perspective-neutral'' framework~\cite{HoehnSubsystem,CastroRuizOreshkov,DGHLM21,Vanrietvelde,Vanrietvelde2,HoehnUniverse,Vanrietvelde3,Trinity,HSLRel,DGHLM21,HoehnKrummMueller} for which the Page-Wootters mechanism~\cite{PageWootters83,Wootters84,Giovannetti,Trinity,HSLRel} (PWM) is a well-known special case, see Fig.~\ref{fig_overlap}: the global system is in a state $|\psi\rangle_{RS}$ that is invariant under symmetry transformations, but conditioning on a subsystem $R$ defines a state $|\psi\rangle_{S|R}$ of $S$ relative to $R$ that is asymmetric. For the PWM in particular, this has led to the slogan of ``time replaced by quantum correlations'', but despite partial results~\cite{Mendes,Martinelli,Carmo}, a quantitative relation between entanglement and conditional time asymmetry has so far not been established.

Here we provide a rigorous formulation of such a relation, and do so in the general case of arbitrary compact Lie symmetry groups $\mathcal{G}$ (not just time translations like in PWM), i.e.\ for the general ``perspective-neutral'' approach that overlaps with constraint quantization. We show an exact correspondence between the amount of entanglement in the global state of $RS$ and the amount of asymmetry in the conditional state of $S$, and demonstrate several resulting insights into the structure and physical properties of such quantum reference frames.

\begin{figure}[H]
\centering \includegraphics[width=.85\columnwidth]{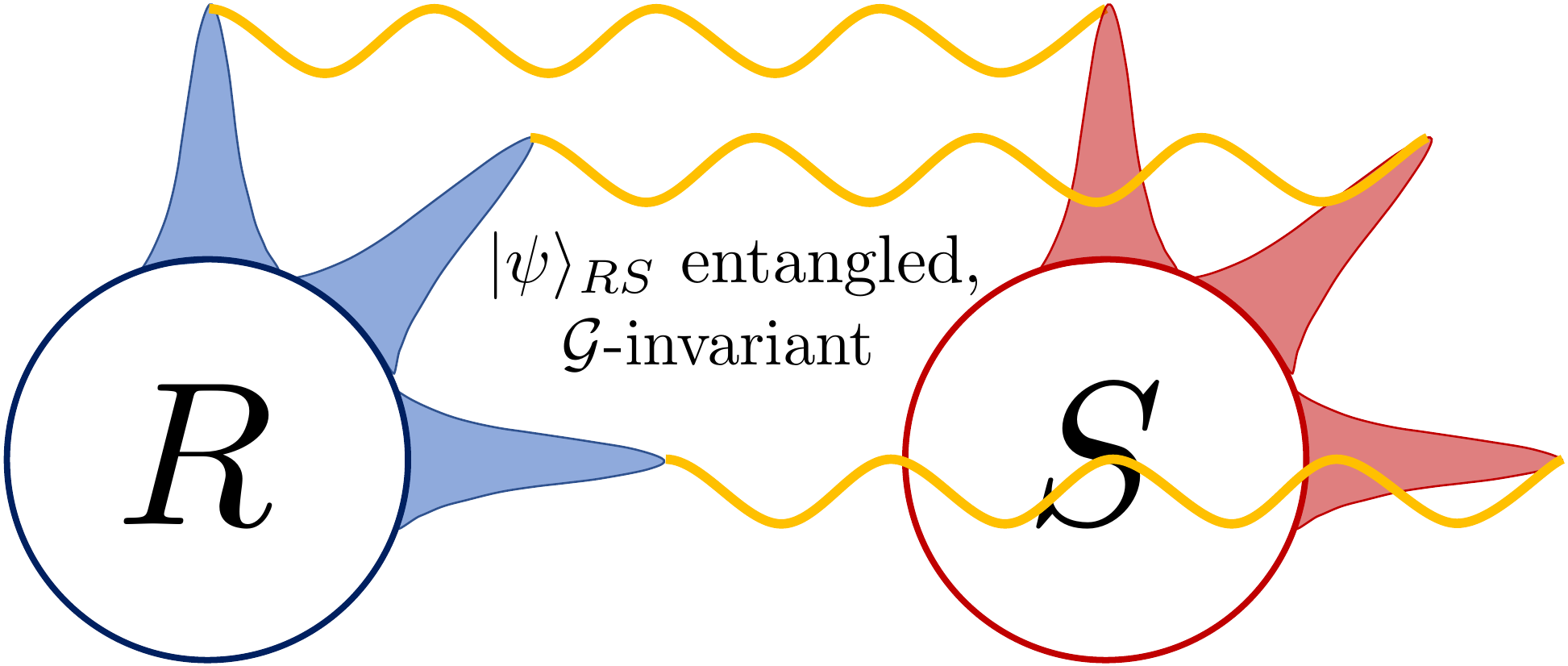}
\includegraphics[width=.85\columnwidth]{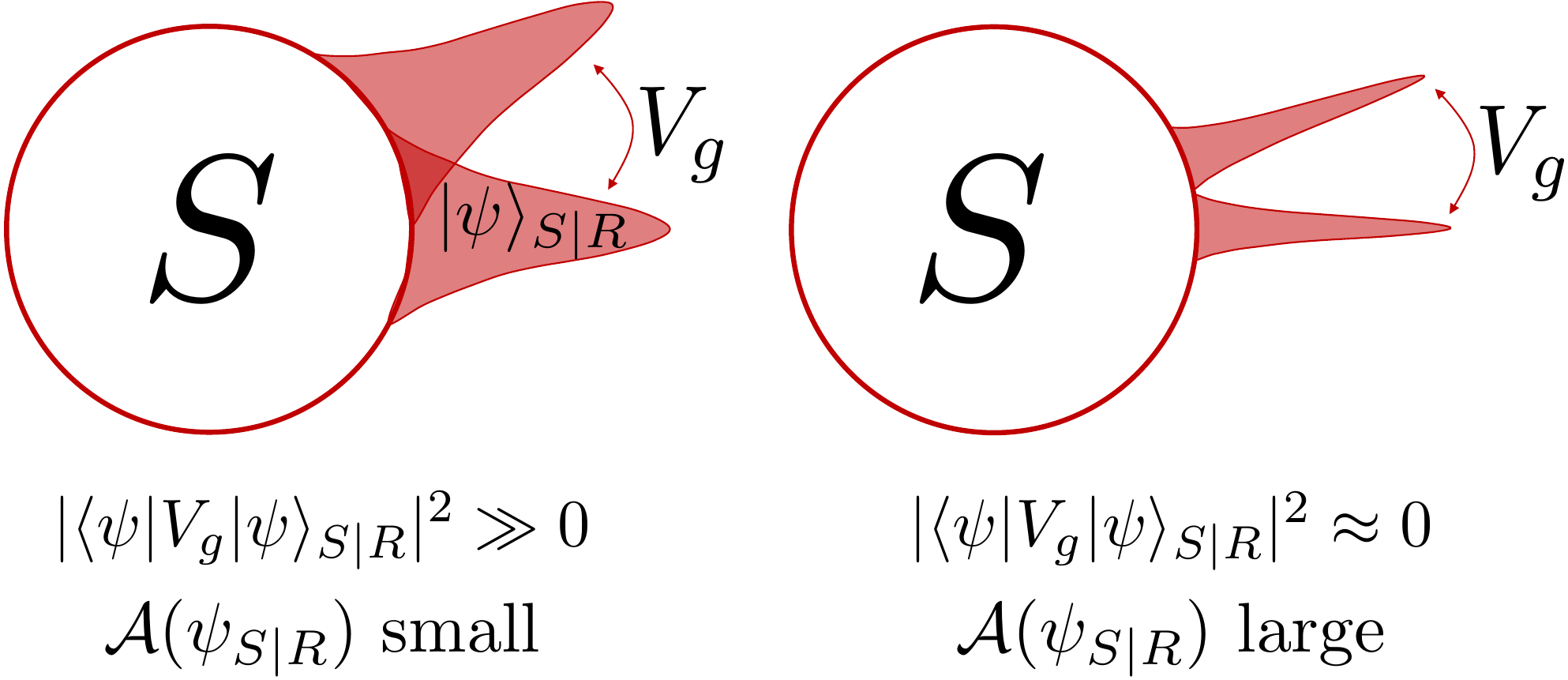}
\caption{\textbf{Upper pane:} a globally invariant (e.g.\ timeless) state $|\psi\rangle_{RS}$ induces asymmetry in a subsystem $S$ by conditioning on the reference frame (e.g.\ clock) $R$. \textbf{Lower pane:} more induced asymmetry amounts to smaller overlap of the conditional state with its translations, i.e.\ a larger value of $\mathcal{A}(\psi_{S|R})$. We prove that the R\'enyi-$2$ entanglement entropy of $\psi_{RS}$ equals the asymmetry of the conditional state $\psi_{S|R}$.}
\label{fig_overlap}
\end{figure}

\emph{Framework.} We consider two quantum systems $R$ (the reference) and $S$ (the system) with $\dim\,R<\infty$, carrying unitary representations $U$ and $V$ of the compact (possibly finite) symmetry group $\mathcal{G}$. That is, $g\in\mathcal{G}$ acts via $U_g\otimes V_g$ on $RS$, and we do not assume that $U$ or $V$ are irreducible. We study states $|\psi\rangle_{RS}$ that are globally invariant, $|\psi\rangle_{RS}=U_g\otimes V_g|\psi\rangle_{RS}$ for all $g$, reflecting that, in a background-independent theory, all physically meaningful properties are purely relational~\cite{Vanrietvelde,Vanrietvelde2,HoehnUniverse,Vanrietvelde3,Trinity,HSLRel,DGHLM21}. Demanding that pure states do not pick up a global phase under the action of $\cg$ can be motivated by preserving entanglement with a purifying system~\cite{HoehnKrummMueller}. Using terminology from constraint quantization, the subspace of all globally invariant states will be called the ``physical Hilbert space'' $\ch_{\rm phys}$~\cite{HenneauxTeitelboim,Vanrietvelde}.

In the PWM, $\mathcal{G}$ is the group of time translations $\mathbb{R}$, represented on $S$ via $V_t=\exp(-it\hat H_S/\hbar)$, and on the (infinite-dimensional) reference $R$ via $U_t |s\rangle_R=|s+t\rangle_R$, where $|s\rangle_R$ denotes an improper clock eigenstate. Formally, the physical Hilbert space consists of the globally ``timeless'' states $|\psi\rangle_{RS}=\int_{\mathbb{R}}dt\,|t\rangle_R\otimes|\psi(t)\rangle_S$.

Now, we are interested in the quantum state of $S$ conditional on $R$ being ``oriented in some direction'' $g\in\mathcal{G}$. To make precise sense of this intuition, we define a coherent state system~\cite{Perelomov,DGHLM21} $\{|g\rangle_R\}_{g\in\mathcal{G}}$ by choosing some normalized state $|e\rangle_R$, where $e$ is the unit element of $\mathcal{G}$, and setting $|g\rangle_R:=U_g |e\rangle_R$. It follows that $U_g |g'\rangle_R = |gg'\rangle_R$. We demand that $|e\rangle_R$ is chosen such that we obtain a resolution of the identity, $\int_{\mathcal{G}}dg\,|g\rangle\langle g|=c\cdot\mathbf{1}_R$ with $c>0$ some constant, where $dg$ denotes the Haar measure on $\mathcal{G}$ (it follows that $c=1/d_R$, with $d_R=\dim\, R$). If $g\mapsto U_g$ is an irreducible representation (irrep), then this follows automatically from Schur's lemma; otherwise it imposes some conditions on $|e\rangle_R$ described in~\cite{DGHLM21}.

The coherent state system defines a covariant positive-operator valued measure (POVM)~\cite{Holevo,Busch,BRWS} that we can use to measure the ``orientation'' $g$ of $R$. If we do this on a state $|\psi\rangle_{RS}\in\mathcal{H}_{\rm phys}$ and find outcome $g$, then the post-measurement state (see Supplemental Material~\ref{SecSomeLemmas}) of $S$ will be
\begin{equation}
   |\psi(g)\rangle_{S|R}:=\sqrt{d_R}\langle g|_R\otimes\mathbf{1}_S |\psi\rangle_{RS}.
   \label{eqConditional}
\end{equation}
We will abbreviate $|\psi\rangle_{S|R}:=|\psi(e)\rangle_{S|R}$, and we will sometimes emphasize the dependence of this state on the choice of $|e\rangle_R$ by writing $|\psi\rangle_{S|R}^{|e\rangle}$. We have~\cite{DGHLM21}
\[
   |\psi(g)\rangle_{S|R}=\sqrt{d_R} \langle e|_R U_g^\dagger\otimes V_g V_g^\dagger |\psi\rangle_{RS}=V_g |\psi(e)\rangle_{S|R}.
\]
This is analogous to the PWM, where conditioning on the state $|t\rangle_R$, i.e.\ on the clock showing time $t$, gives us the time-evolved state $|\psi(t)\rangle_S=|\psi(t)\rangle_{S|R}=V_t |\psi(0)\rangle_{S|R}$.

\emph{Conditional asymmetry.} While the initial physical state $|\psi\rangle_{RS}$ is fully symmetric, we would like the measurement of $R$ to break the symmetry and to lead to an \emph{asymmetric} conditional state, i.e.\ $|\psi\rangle_{S|R} \neq V_g |\psi\rangle_{S|R}$ for $g\neq e$. Intuitively, for a ``good'' reference frame $R$, it should be possible to locate the system $S$ very precisely relative to $R$. For example, if $\mathcal{G}$ is a finite group and $S$ carries an orthonormal coherent state system $\{|g\rangle_S\}_{g\in\mathcal{G}}$, this could mean that $|\psi\rangle_{S|R}$ is very strongly peaked on a single $h\in\cg$, i.e.\ $|\langle h|\psi\rangle_{S|R}|\approx 1$. Then, translating by $g\neq e$ will lead to a state $|\psi(g)\rangle_{S|R}$ almost orthogonal to $|\psi\rangle_{S|R}$, because $0\approx |\langle g^{-1}h|\psi\rangle_{S|R}|=|\langle h|\psi(g)\rangle_{S|R}|$.

In other words, a ``good'' reference frame $R$ should lead to a conditional state $|\psi\rangle_{S|R}=|\psi(e)\rangle_{S|R}$ that is well-distinguishable from its ``translations'' $V_g |\psi\rangle_{S|R}=|\psi(g)\rangle_{S|R}$, see Fig.~\ref{fig_overlap}. A well-known operational quantifier of distinguishability of quantum states is the \emph{fidelity}~\cite{NielsenChuang}, $\mathcal{F}(\rho,\sigma):=({\rm tr}\sqrt{\sqrt{\rho}\sigma\sqrt{\rho}})^2$. We have $0\leq\mathcal{F}(\rho,\sigma)\leq 1$, where $0$ is attained if and only if $\rho$ and $\sigma$ are perfectly distinguishable, and $1$ if and only if $\rho=\sigma$. For pure states, it reduces to $\mathcal{F}(|\psi\rangle,|\varphi\rangle)=|\langle\psi|\varphi\rangle|^2$. This motivates the following definition.
\begin{definition}
Given any physical state $|\psi\rangle_{RS}\in\mathcal{H}_{\rm phys}$, the \emph{conditional uniformity} of the corresponding conditional state $|\psi\rangle_{S|R}\equiv |\psi\rangle_{S|R}^{|e\rangle}$ (as defined in Eq.~(\ref{eqConditional})) is
\begin{equation}
   \mathcal{U}(\psi_{S|R}):=
   \int_{\mathcal{G}}dg\, \mathcal{F}\left(\strut|\psi\rangle_{S|R},V_g|\psi\rangle_{S|R}\right),
   \label{EqDef1}
\end{equation}
and we define its \emph{conditional asymmetry} as $\mathcal{A}(\psi_{S|R}):=-\log \mathcal{U}(\psi_{S|R})$.
\end{definition}
Intuitively, for ``bad'' quantum reference frames $R$, $|\psi\rangle_{S|R}\approx V_g|\psi\rangle_{S|R}=|\psi(g)\rangle_{S|R}$ for many $g$, and $\mathcal{U}(\psi_{S|R})$ will be close to unity; and for ``good'' ones, this quantity will be close to zero. By invariance of the Haar measure, the conditional uniformity is the same for all $|\psi(g)\rangle_{S|R}$ and can also be written
\begin{equation}
   \mathcal{U}(\psi(g)_{S|R})
   =\int_{\mathcal{G}}dg'\,\int_{\mathcal{G}}dg''\, |\langle\psi(g')|\psi(g'')\rangle_{S|R}|^2.
   \label{eqPower}
\end{equation}
A priori, conditional uniformity will depend on the choice of coherent state system $\{|g\rangle_R\}_{g\in\mathcal{G}}$ since $|\psi\rangle_{S|R}=|\psi\rangle_{S|R}^{|e\rangle}$ does. Intuitively, the choice of covariant POVM that is used to measure the reference frame should have some impact on the quality of its use. Surprisingly, however, this intuition does not hold up in our context. Using the notation $\mathcal{H}_\alpha(\rho):=\frac 1 {1-\alpha}\log{\rm tr}(\rho^\alpha)$ for the R\'enyi-$\alpha$ entropy of a quantum state $\rho$, we get:
\begin{theorem}
\label{TheA}
The conditional asymmetry of $\psi_{S|R}$ equals the R\'enyi-2 entanglement entropy of $\psi_{RS}$:
\[
   \mathcal{A}(\psi_{S|R})=\mathcal{H}_2({\rm Tr}_R |\psi\rangle\langle\psi|_{RS}).
\]
In particular, $\mathcal{A}(\psi_{S|R})\equiv \mathcal{A}(\psi_{S|R}^{|e\rangle})$ is \emph{independent} of the choice of coherent state system, i.e.\ of $|e\rangle$, and can be understood as a function of the physical state $|\psi\rangle_{RS}$.
\end{theorem}
\emph{Proof.} Expanding the definition of $\mathcal{U}$, we find
\begin{widetext}
\begin{eqnarray*}
\mathcal{U}(\psi_{S|R})&=&d_R^2\int_{\mathcal{G}}dg\,|\langle\psi|_{RS}(|e\rangle_S\otimes\mathbf{1}_S)(\langle g|_R\otimes\mathbf{1}_S)|\psi\rangle_{RS}|^2\\
&=& d_R^2 \langle\psi|_{RS}(|e\rangle_R\otimes\mathbf{1}_S)\int_{\mathcal{G}}dg\, (\langle g|_R\otimes\mathbf{1}_S)|\psi\rangle\langle\psi|_{RS}(|g\rangle_R\otimes\mathbf{1}_S)(\langle e|_R\otimes\mathbf{1}_S)|\psi\rangle_{RS}\\
&=& d_R \langle \psi|_{RS}\left(|e\rangle\langle e|_R\otimes {\rm Tr}_R |\psi\rangle\langle\psi|_{RS}\right)|\psi\rangle_{RS},
\end{eqnarray*}
\end{widetext}
where we have used that the $|g\rangle_R$ yield a resolution of the identity. To simplify this further, replace both occurrences of $|\psi\rangle_{RS}$ by $U_g^\dagger\otimes V_g^\dagger|\psi\rangle_{RS}$, use that $V_g {\rm Tr}_R |\psi\rangle\langle\psi|_{RS} V_g^\dagger={\rm Tr}_R |\psi\rangle\langle\psi|_{RS}$, and compute the group average over $g$ of the resulting expression. This yields $\mathcal{U}(\psi_{S|R})=\langle\psi|_{RS}(\mathbf{1}_R\otimes{\rm Tr}_R|\psi\rangle\langle\psi|_{RS})|\psi\rangle_{RS}={\rm tr}[({\rm Tr}_R|\psi\rangle\langle\psi|_{RS})^2]$ which is independent of $|e\rangle_R$. Finally, take the negative logarithm of both sides.
\qed

In the special case where $\mathcal{G}$ is a finite subgroup of time translations, a version of this result has been given in~\cite{Boette}.

\emph{Resource-theoretic consequences.} As we have just seen, for conditional asymmetry, any choice of coherent state system $\{|g\rangle_R\}_{g\in\mathcal{G}}$ is as good as any other, $\{|g\rangle'_R\}_{g\in\mathcal{G}}$: we have $\mathcal{A}(\psi_{S|R}^{|e\rangle})=\mathcal{A}(\psi_{S|R}^{|e\rangle'})$ for all $|\psi\rangle_{RS}\in\mathcal{H}_{\rm phys}$. However, $\mathcal{A}$ is just \emph{one possible} measure of asymmetry. Independence from $|e\rangle_R$ does not hold for all possible measures of asymmetry. For example, taking the $4$th power instead of the $2$nd in Eq.~(\ref{eqPower}) defines an alternative measure of uniformity that \emph{does} depend on the choice of $|e\rangle$, see Supplemental Material~\ref{app:altmeasures} for an example.

Thus, a more systematic and operational approach is warranted. Such an approach is to study asymmetry in the context of a \emph{resource theory}~\cite{ChitambarGour}. Resource theories provide useful tools to describe and quantify the role that a resource plays in the performance of certain tasks, be it in thermodynamics~\cite{Brandao} or entanglement theory~\cite{Vedral}. Here, we use the resource theory of asymmetry~\cite{Marvian,MarvianSpekkens} to study the quality of a quantum reference frame.

To this end, let us introduce several relevant notions. Consider any representation $g\mapsto U_g$ of a compact group $\mathcal{G}$. A quantum state $\rho$ is \emph{$\mathcal{G}$-invariant} if $U_g\rho U_g^\dagger = \rho$ for all $g\in\mathcal{G}$. We say that a quantum operation $\mathcal{E}$ on the operators of the corresponding Hilbert space is \emph{$\mathcal{G}$-covariant} if $\mathcal{E}(U_g \rho U_g^\dagger)=U_g \mathcal{E}(\rho) U_g^\dagger$ for all $\rho$ and all $g\in\mathcal{G}$. Crucially, $\mathcal{G}$-covariant operations map $\mathcal{G}$-invariant states to $\mathcal{G}$-invariant states --- in this sense, they cannot create $\mathcal{G}$-asymmetry. We say that $\rho$ is \emph{at least as asymmetric as $\rho'$} if there is a $\mathcal{G}$-covariant operation $\mathcal{E}$ with $\mathcal{E}(\rho)=\rho'$. If also $\rho'$ is at least as asymmetric as $\rho$, we say that $\rho$ and $\rho'$ are \emph{equally $\mathcal{G}$-asymmetric}. Note that there are also pairs of states with the property that neither one is at least as asymmetric as the other. In this case, we say that $\rho$ and $\rho'$ are \emph{incomparably $\mathcal{G}$-asymmetric}.

Our result $\mathcal{A}(\psi_{S|R}^{|e\rangle})=\mathcal{A}(\psi_{S|R}^{|e\rangle'})$ does not automatically imply that $\psi_{S|R}^{|e\rangle}$ and $\psi_{S|R}^{|e\rangle'}$ are equally $\mathcal{G}$-asymmetric, but we can show the following:
\begin{theorem}
\label{TheResource}
Consider two choices of coherent state system, $\{|g\rangle_R\}_{g\in\mathcal{G}}$ and $\{|g\rangle'_R\}_{g\in\mathcal{G}}$. Then, for every physical state $|\psi\rangle_{RS}\in\mathcal{H}_{\rm phys}$, the conditional states $|\psi\rangle_{S|R}^{|e\rangle}$ and $|\psi\rangle_{S|R}^{|e\rangle'}$ are either equally or incomparably $\mathcal{G}$-asymmetric. That is, in the resource-theoretic sense, no coherent state system induces more asymmetry on $S$ than any other.
\end{theorem}
The proof is given in Supplemental Material~\ref{app:proof2}. It employs techniques of \cite{Marvian,MarvianSpekkens} which use \emph{characteristic functions} $\chi_\varphi(g):=\langle\varphi|V_g|\varphi\rangle$ to characterize pure-state interconvertibility under $\mathcal{G}$-covariant operations. They are related to conditional uniformity via $\mathcal{U}(\psi_{S|R})=\int_{\mathcal{G}}dg\,|\chi_{\psi_{S|R}}(g)|^2$.

\emph{Physical Hilbert space average.} To quantify how much asymmetry the quantum reference frame $R$ is able to induce on $S$, we have to go beyond single conditional states and consider the collection of all $|\psi\rangle_{S|R}$. Theorem~\ref{TheA} allows us to do so in a particularly elegant way: conditional uniformity $\mathcal{U}(\psi_{S|R}^{|e\rangle})$ is independent of the choice of seed coherent state $|e\rangle$ and can be understood as a function of $|\psi\rangle_{RS}$. We can thus determine an average of this quantity over all conditional states by computing the Hilbert space average of $\mathcal{U}$ over $\mathcal{H}_{\rm phys}$. Not only can this be done analytically, but the result will be independent of the coherent state system and quantify the quality of the reference frame in terms of simple properties of the representations $g\mapsto U_g$ on $R$ and $g\mapsto V_g$ on $S$.

To this end, let us fix some notation. Following~\cite{Simon}, the set of (unitarily inequivalent) irreps of $\mathcal{G}$ is denoted $\mathcal{\hat G}$. Decomposing the representation on $R$ into irreps, we get $U_g=\bigoplus_{\alpha\in\mathcal{\hat G}}n^U_\alpha T_g^{(\alpha)}$, where $n_\alpha^U\in \mathbb{N}_0$ is the multiplicity of the irrep $T^{(\alpha)}$, $g\mapsto T_g^{(\alpha)}$. Similarly, $V_g=\bigoplus_{\beta\in\mathcal{\hat G}}n_\beta^V T_g^{(\beta)}$. The dimension of the irrep $\alpha$ will be denoted $d_\alpha$, and the conjugate representation of $\alpha$ will be denoted $\bar\alpha$, i.e.\ $T_g^{(\bar\alpha)}=\overline{T_g^{(\alpha)}}$ in some basis. Note that the existence of a coherent state system on $R$ furnishing a resolution of the identity implies that $n_\alpha^U\leq d_\alpha$~\cite{DGHLM21}. The unitarily invariant measure on the unit vectors of $\mathcal{H}_{\rm phys}$, normalized such that $\int_{\mathcal{H}_{\rm phys}}d\psi=1$, will be denoted $d\psi$.
\begin{theorem}
\label{TheAvg}
The physical Hilbert space average $\mathcal{U}_{\rm phys}$ of conditional uniformity (``physical uniformity'') is
\[
   \int_{\mathcal{H}_{\rm phys}} d\psi\,\mathcal{U}(\psi_{S|R})=\frac 1 {d_{\rm phys}(d_{\rm phys}+1)}\sum_{\alpha\in\mathcal{\hat G}} \frac{n_\alpha^U n_{\bar\alpha}^V(n_\alpha^U+n_{\bar\alpha}^V)}{d_\alpha},
\]
where $d_{\rm phys}=\dim\,\mathcal{H}_{\rm phys}=\sum_{\alpha\in\mathcal{\hat G}} n_\alpha^U n_{\bar\alpha}^V$.
\end{theorem}
The proof is given in Supplemental Material~\ref{app:proof3}. It relies on quantum information techniques to compute Hilbert space averages of polynomials via the replica trick, together with representation-theoretic orthogonality and convolution identities for the characters of the group $\mathcal{G}$.

Like $\mathcal{U}(\psi_{S|R})$, the physical uniformity $\mathcal{U}_{\rm phys}$ lies in the interval $(0,1]$, and it is independent of the choice of coherent state system $\{|g\rangle_R\}_{g\in\mathcal{G}}$. Its value, or rather that of the corresponding \emph{physical asymmetry} $\mathcal{A}_{\rm phys}:=-\log \mathcal{U}_{\rm phys}$, quantifies in representation-theoretic terms how much asymmetry $R$ induces on $S$ on average. Since entanglement entropy is upper-bounded by $\log d_R$, Theorem~\ref{TheA} tells us that $\mathcal{U}(\psi_{S|R})\geq 1/d_R$ for all $|\psi_{RS}\rangle\in\mathcal{H}_{\rm phys}$, and thus $\mathcal{A}_{\rm phys}\leq \log d_R$. This shows directly the \emph{necessity of large Hilbert space dimension} for a quantum reference frame $R$ to induce large amounts of asymmetry.

\emph{Example: maximum spin-$J$ reference frame.} Consider a quantum reference frame for the group ${\rm SU}(2)$, where we constrain $R$ to contain only irreps of spin $J$ or less, i.e.\ irreps labelled by $\alpha=0,\frac 1 2,1,\ldots,J$. As stated above Theorem \ref{TheAvg}, the resolution of the identity necessarily implies $n_\alpha^U\leq d_\alpha$. Thus, the best we can have is equality, i.e.\ $n_\alpha^U=d_\alpha=2\alpha+1$. Suppose that the system of interest is $S=L^2({\rm SU}(2))$, i.e.\ the infinite-dimensional system of wave functions on the group. By the Peter-Weyl theorem~\cite{Simon}, we have $n_\alpha^V=d_\alpha$. The physical Hilbert space has dimension $d_{\rm phys}=d_R=\sum_{k=0}^{2J} n_{k/2}^U n_{k/2}^V=\sum_{k=0}^{2J}(k+1)^2\sim 8  J^3/3$. The physical uniformity evaluates to $\mathcal{U}_{\rm phys}=2/(d_{\rm phys}+1)\sim 3 J^{-3}/4$. In particular, for $J\to\infty$, the conditional states $|\psi\rangle_{S|R}$ become, on average, perfectly distinguishable from their rotated versions $|\psi(g)\rangle_{S|R}$. This indicates that, for increasing $J$, this $R$ resembles more and more a perfect ``classical'' reference frame, and the above tells us how good this approximation is for finite $J$.

Here, physical asymmetry is close to its maximal value: $\mathcal{A}_{\rm phys}=\log d_R - \log 2 + \mathcal{O}(J^{-3})$. This implies that the average asymmetry $\overline{\mathcal{A}_{\rm phys}}:=\int_{\mathcal{H}_{\rm phys}}d\psi\, \mathcal{A}(\psi_{S|R})$ is also large. The latter follows because $(-\log)$ is convex, and therefore Jensen's inequality tells us that $\overline{\mathcal{A}_{\rm phys}}=\int_{\mathcal{H}_{\rm phys}}d\psi\,(-\log \mathcal{U}(\psi_{S|R}))\geq -\log \int_{\mathcal{H}_{\rm phys}}d\psi\,\mathcal{U}(\psi_{S|R})=\mathcal{A}_{\rm phys}$. However, if the average asymmetry is close to maximal, then \emph{most conditional states $|\psi\rangle_{S|R}$ must be almost maximally asymmetric}. This is reminiscent of the phenomenon in quantum information theory that almost all pure states are almost maximally entangled~\cite{HaydenLeungWinter}, and we can exploit this analogy rigorously via the correspondence of Theorem~\ref{TheA}.

\emph{Typical asymmetry.} Suppose that we pick a state $|\psi\rangle_{RS}$ at random from $\mathcal{H}_{\rm phys}$ according to the unitarily invariant measure. Then the techniques of~\cite{HaydenLeungWinter} (in particular Lemmas III.1 and III.8) together with the Lipschitz bound on entanglement entropy $\mathcal{H}_2$ from~\cite{HaydenWinter}  imply that
\[
   {\rm Prob}\left\{|\mathcal{A}(\psi_{S|R})-\overline{\mathcal{A}_{\rm phys}}|\geq\varepsilon\right\}\leq 2\exp\left(-\frac{d_{\rm phys}\varepsilon^2}{C\sqrt{d_R}}\right)
\]
for all $\varepsilon\geq 0$, where $C=72\pi^3\log 2$. Using $\overline{\mathcal{A}_{\rm phys}}\geq \mathcal{A}_{\rm phys}$, the right-hand side also upper-bounds the probability that $\mathcal{A}(\psi_{S|R})\leq\mathcal{A}_{\rm phys}-\varepsilon$. In the ${\rm SU}(2)$-example above, for every fixed $\varepsilon>0$, the probability that a random $|\psi\rangle_{S|R}$ has conditional asymmetry less than $\log d_R-\log 2-\varepsilon$ is exponentially small in $J^{3/2}$: almost all conditional states are indeed almost maximally asymmetric.

\emph{Example: periodic quantum clock.} Consider a harmonic oscillator $S$ with Hamiltonian $\hat H_S=\frac{\hat p^2}{2m}+\frac 1 2 m \omega^2 \hat x^2=\hbar \omega \sum_{n=0}^\infty \left(n+\frac 1 2\right)|n\rangle\langle n|$. It evolves periodically in time with period $\tau=4\pi/\omega$, hence we can measure time with a periodic quantum clock $R$, carrying a representation $U_g$ of $\mathcal{G}={\rm U}(1)\simeq [0,2\pi)$. This generalizes the notion of clocks used in the PWM with $\mathcal{G}=\mathbb{R}$ to the periodic case. The representations of ${\rm U}(1)$ are labelled by integers $\alpha\in\mathbb{Z}$, with $T^{(\alpha)}_g=\exp(i\alpha g)$ and $d_\alpha=1$. Suppose that $R$ is an imperfect clock, with $U_g=\bigoplus_{\alpha=-k}^k e^{i\alpha g}|\alpha\rangle\langle\alpha|$ where $k\geq 1$ is odd and finite. Then $n_\alpha^U=1$ if $-k\leq \alpha\leq k$ and $0$ otherwise. For $k\to\infty$, one would recover a perfect clock, i.e.\ a periodic version of the PWM, with associated Hilbert space $L^2({\rm U}(1))$.

Up to arbitrary changes of phase in the one-dimensional subspaces spanned by $|\alpha\rangle$, the unique choice of coherent state system yielding a resolution of the identity is via $|e\rangle_R=\frac 1 {\sqrt{d_R}}\sum_{\alpha=-k}^k |\alpha\rangle$ with $d_R=2k+1$. Interestingly, in a scenario where we are given a quantum clock and would like to determine time as accurately as possible by measuring the clock, this state has been shown by Holevo~\cite{Holevo} to generate the POVM which is optimal for a large class of cost functions (see also~\cite{BuzekDerkaMassar}). However, Holevo's results are not directly applicable to our scenario, since our clock is in a joint stationary state $|\psi\rangle_{RS}\in\mathcal{H}_{\rm phys}$ with the harmonic oscillator.

Associating time $t$ with the group element $g\in[0,2\pi)$ via $g=2\pi t/\tau$, we get that $V_g=\exp\left(-i\frac{\tau g}{2\pi\hbar}\hat H_S\right)$ represents time translations on $S$. On the energy eigenstates, we have $V_g|n\rangle=\exp(-ig(2n+1))|n\rangle$, hence $n_\alpha^V=1$ if $\alpha=-1,-3,-5,\ldots$ and $0$ otherwise. We have
\begin{equation}
   \mathcal{H}_{\rm phys}={\rm span}\left\{|2n+1\rangle_R|n\rangle_S\,\,|\,\, 0\leq n\leq(k-1)/2\right\},
   \label{eqPhys}
\end{equation}
and so $d_{\rm phys}=(k+1)/2$. The physical uniformity becomes $\mathcal{U}_{\rm phys}=4/(k+3)$, which is the physical Hilbert space average of $\mathcal{U}(\psi_{S|R})=\int_0^\tau \frac{dt}\tau |\langle \tilde\psi_{S|R}(0)|\tilde\psi_{S|R}(t)\rangle|^2$, with the time-evolved state $\tilde\psi(t):=\psi(2\pi t/\tau)$. This can be interpreted as an instance of a \emph{time-energy uncertainty relation}: the larger the range of energies in the clock $R$ (i.e.\ the larger $k$), the more distinguishable will the conditional states be from their time-translated versions on average.

Compare this with what we can learn from the the Mandelstam-Tamm quantum speed limit~\cite{MandelstamTamm,AnandanAharonov,Hoernedal,MarvianSpekkensZanardi}: $|\langle\tilde\psi_{S|R}(0)|\tilde\psi_{S|R}(t)\rangle|^2\geq \cos_*^2(\Delta \hat H_S\cdot t/\hbar)$, where $\cos_*(x)=\cos x$ if $0\leq x \leq \pi/2$ and $0$ otherwise. This implies
\begin{eqnarray*}
\mathcal{U}(\psi_{S|R})&=&\int_0^\tau \frac {dt}\tau |\langle\tilde\psi_{S|R}(0)|\tilde\psi_{S|R}(t)\rangle|^2\\
&\geq& \int_0^{\min\{t_0,\tau\}}\frac{dt}\tau\cos^2\left(\frac{t\Delta\hat H_S}\hbar\right),
\end{eqnarray*}
where $t_0=\pi\hbar/(2\Delta\hat H_S)$. Due to~(\ref{eqPhys}), every $\psi_{S|R}$ is supported on the subspace spanned by $|0\rangle_S,\ldots,|(k-1)/2\rangle_S$, hence these states have $\Delta\hat H_S\leq\hbar\omega k/2$. Most states $\psi_{S|R}$ will have $\Delta\hat H_S\geq \hbar\omega/8$ if $k$ is large, hence $t_0\leq\tau$. For those $\psi_{S|R}$, the integral evaluates to $\mathcal{U}(\psi_{S|R})\geq\omega\hbar/(16\Delta\hat H_S)\geq 1/(8k)$. Our result on $\mathcal{U}_{\rm phys}$ gives essentially the same scaling in $k$, but improves on the Mandelstam-Tamm result by a factor of $32$ on average. Hence, Theorem~\ref{TheAvg} provides a representation-theoretic time-energy trade-off and generalizes it to more general groups than time translations.

\emph{Conclusions.} We have shown that there is an exact quantitative correspondence between the amount of entanglement in a globally symmetric quantum state and the amount of asymmetry in the conditional state relative to an internal quantum reference frame, leading to several insights on the quality of imperfect reference frames, speed limits, and typicality of asymmetry. We have also begun to explore the resource-theoretic consequences of our duality in Theorem~\ref{TheResource}, using the close relation between conditional uniformity and characteristic functions. It would be interesting to explore further how resource-theoretic notions can be imported into this ``perspective-neutral'' scenario. It would also be worthwhile to explore the generalization to non-compact groups~\cite{DGHLM21} such as the Lorentz and Galilei groups. These possible extensions notwithstanding, we believe that our results shed significant light on the quantum information-theoretic and structural foundations of internal quantum reference frames.\\

\emph{Acknowledgments.} We are grateful to Philipp A.\ H\"ohn for helpful discussions. AD was supported by the Austrian Science Fund (FWF) through BeyondC (F7103-N48), the John Templeton Foundation (ID\# 61466) as part of The Quantum Information Structure of Spacetime (QISS) Project (qiss.fr) and the European Commission via Testing the Large-Scale Limit of Quantum Mechanics (TEQ) (No. 766900) project. SLL and MPM acknowledge support from the Austrian Science Fund (FWF) via project P 33730-N.

\onecolumngrid

\section*{Supplemental Material}
In the following, we will assume that not only $R$ but also $S$ is finite-dimensional. In the main text, we have been discussing cases where $\dim\, S=\infty$, but in all these cases, it is possible to restrict to a finite-dimensional subspace of $S$ due to $d_{\rm phys}<\infty$. Hence the assumption $\dim\, S<\infty$ is no loss of generality.

\section{Some definitions and lemmas}
\label{SecSomeLemmas}
We begin with a definition of \emph{positive definite function} used in the proof of Theorem~\ref{TheResource}. Our exposition follows~\cite{Marvian}.
\begin{definition}
Let $\mathcal{G}$ be a compact Lie group. A continuous function $f:\mathcal{G}\to\mathbb{C}$ is called \emph{positive definite} if
\[
   \int_{\mathcal{G}}dg\int_{\mathcal{G}} dh\, \overline{\varphi(g)}f(g^{-1}h)\varphi(h)\geq 0
\]
for all $\varphi\in L^1(\mathcal{G})$.
\end{definition}
Positive definite functions are also called \emph{functions of positive type}, see~\cite[Sec.\ 3.3]{Folland} for further mathematical details.

Clearly, $\{\bar\varphi\,\,|\,\,\varphi\in L^1(\mathcal{G})\}=L^1(\mathcal{G})$, and so taking the complex conjugate of the inequality above shows that if $f$ is positive definite then so is $\bar f$. This is also the content of Proposition 3.14 in~\cite{Folland}.

\begin{lemma}
\label{LemNormalized}
The \enquote{reduction map} $R^{(g)}: \ch_\phys \to S,\  R^{(g)}= \sqrt{d_R}\bra{g}_R \otimes \mathbf{1}_S$, is an isometry.
\end{lemma}
\begin{proof}
Since $|\psi(g)\rangle_S=V_g |\psi(e)\rangle_S$, it is sufficient to consider the case $g=e$. This can be done as follows:
\begin{eqnarray*}
   \langle R^{(e)}\psi|R^{(e)}\phi\rangle_S&=&d_R \int_{\mathcal{G}}dg\, \langle\psi|_{RS} U_g\otimes V_g (|e\rangle_R\otimes\mathbf{1}_S)(\langle e|_R\otimes\mathbf{1}_S)U_g^\dagger\otimes V_g^\dagger |\phi\rangle_{RS}\\
   &=&d_R \langle\psi|_{RS}\left(\int_{\mathcal{G}} dg\,|g\rangle\langle g|_R\right) |\phi\rangle_{RS}=\bracket{\psi}{\phi}_{RS},
\end{eqnarray*}
where we have used the invariance $|\psi\rangle_{RS}=U_g^\dagger\otimes V_g^\dagger |\psi\rangle_{RS}$ for $|\psi\rangle_{RS}\in\mathcal{H}_{\rm phys}$.
\end{proof}
Note that Lemma \ref{LemNormalized} was already shown to hold in Ref.~\cite{DGHLM21}.

\begin{lemma}
The coherent state system $\{\ket{g}_R\}_{g\in\mathcal{G}}$ can be used to calculate the partial trace $\Tr_R$ as follows:
\begin{equation*}
\Tr_R\left(\rho_{RS}\right)=d_R\int_{\cg} dg \bra{g}_R \otimes \mathbf{1}_S \rho_{RS}\ket{g}_R \otimes \mathbf{1}_S .
\end{equation*}
\label{partraceandg}
\end{lemma}
\begin{proof}
Let us define $\displaystyle\sigma_S\coloneqq d_R\int_{\cg} dg \bra{g}_R \otimes \mathbf{1}_S \rho_{RS}\ket{g}_R \otimes \mathbf{1}_S$. Let $X_S\in\mathcal{B}(S)$ be an arbitrary operator, then
\begin{eqnarray*}
\left(\sigma_S, X_S\right)_S&=&\Tr_S \left(\sigma^{\dagger}_S X_S\right)= d_R\int_{\cg} dg \sum_{i=1}^{d_S} \bra{g}_R \otimes \bra{i}_S \rho^\dagger_{RS}\mathbf{1}\otimes X_S \ket{g}_R \otimes \ket{i}_S \nonumber\\
&=& d_R\int_{\cg} dg\sum_{i=1}^{d_S} \Tr_{RS}\left(\rho^\dagger_{RS}\mathbf{1}_R\otimes X_S \ket{g}\bra{g}_R\otimes \ket{i}\bra{i}_S\right)= \Tr_{RS}\left(\rho^\dagger_{RS}\mathbf{1}_R\otimes X_S \right)\nonumber\\
&=& \Tr_S\left(\left(\Tr_R(\rho_{RS})\right)^\dagger X_S\right)=\left(\Tr_R(\rho_{RS})\ , X_S\right)_S,
\end{eqnarray*}
where $\{\ket{i}_S\}_{i=1}^{d_S}$ is an ONB of $S$ and $\left(\;,\;\right)_S$ denotes the Hilbert-Schmidt inner product of operators.
\end{proof}

\begin{lemma}
Let $\ket{\psi}_{RS}\in\ch_\phys$. The reduced state of $S$ is invariant under the group action: for all $g\in\mathcal{G}$,
\begin{equation*}
\Tr_R\left( \ket{\psi}\bra{\psi}_{RS}\right))=V_g\Tr_R\left( \ket{\psi}\bra{\psi}_{RS}\right))V_g^\dagger.
\end{equation*}
\label{invredstate}
\end{lemma}
\begin{proof}
Let $X_S\in\mathcal{B}(S)$ be an arbitrary operator, then
\begin{eqnarray*}
\left(V_g\Tr_R\left( \ket{\psi}\bra{\psi}_{RS}\right)V_g^\dagger,X_S\right)_S &=&\Tr_S\left(V_g\Tr_R\left( \ket{\psi}\bra{\psi}_{RS}\right)V_g^\dagger X_S\right)
= \Tr_{RS}\left(\ket{\psi}\bra{\psi}_{RS}\cdot\mathbf{1}_R\otimes\left(V_g^\dagger X_S V_g\right)\right)\nonumber\\
&=&\Tr_{RS}\left(U_g^\dagger\otimes V_g^\dagger \cdot\ket{\psi}\bra{\psi}_{RS}\cdot\mathbf{1}_R\otimes X_S \cdot U_g\otimes V_g\right)
=\Tr_{RS}\left(\ket{\psi}\bra{\psi}_{RS}\cdot\mathbf{1}_R\otimes X_S\right)\nonumber\\
&=&\Tr_S\left(\Tr_R\left( \ket{\psi}\bra{\psi}_{RS}\right)X_S\right)
= \left(\Tr_R\left( \ket{\psi}\bra{\psi}_{RS}\right),X_S\right)_S. 
\end{eqnarray*}
In the third line, we used $U^\dagger_g\otimes V^\dagger_g\ket{\psi}_{RS}=\ket{\psi}_{RS}$.
\end{proof}
The flip operator $\mathbb{F}$ on a product Hilbert space $\mathcal{H}_1\otimes\mathcal{H}_2$ is defined by linear extension of $\mathbb{F}|\varphi\rangle\otimes|\psi\rangle=|\psi\rangle\otimes|\varphi\rangle$.
\begin{lemma}
For operators $A$ and $B$, we have $\tr(\mathbb{F} A\otimes B) = \tr (A \cdot B)$.    \label{tensor product flip trick}
\end{lemma} 
\begin{proof}
A simple calculation gives
\begin{align*}
    \tr(\mathbb{F} A\otimes B) &= \tr \left(\mathbb{F} \sum_{i,j,s,k} a_{ij} b_{sk} \ketbra{i}{j} \otimes \ketbra{s}{k}\right)= \tr \left(\sum_{i,j,s,k} a_{ij} b_{sk} \ketbra{s}{j} \otimes \ketbra{i}{k}\right) = \sum_{i,j,s,k} a_{ij} b_{sk} \delta_{sj} \delta_{ik}  \\
    &= \sum_{ij} a_{ij} b_{ji} = \tr (A \cdot B) .
\end{align*}
\vskip -2em
\end{proof}
\begin{lemma}
We have $\tr(\rho_S^2)=\tr(\mathbf{1}_{RR'} \otimes \mathbb{F}_{SS'} (\ketbra{\psi}{\psi}_{RS} \otimes \ketbra{\psi}{\psi}_{R'S'}))$. \label{trace of squared operator}
\end{lemma} 
\begin{proof}
We can easily see that
\begin{align*}
    \tr(\rho_S^2)&= \tr\left(\mathbb{F}_{SS'}\rho_S \otimes \rho_{S'} \right)
    = \tr(\mathbb{F}_{SS'} \Tr_{RR'} (\ketbra{\psi}{\psi}_{RS} \otimes \ketbra{\psi}{\psi}_{R'S'}))=\tr(\mathbf{1}_{RR'} \otimes \mathbb{F}_{SS'} (\ketbra{\psi}{\psi}_{RS} \otimes \ketbra{\psi}{\psi}_{R'S'})).
\end{align*}
\vskip -2em
\end{proof}

\begin{lemma}
The dimension of the physical Hilbert space $\ch_\phys$ is given by
$d_\phys=\sum_\alpha n_\alpha^U n_{\bar{\alpha}}^V$,
where $U$ and $V$ are the representations carried by reference frame $R$ and system $S$ respectively, and $n_\alpha^U$ is the multiplicity of the irrep $T^{(\alpha)}$ in $U$ whereas $n_{\bar{\alpha}}^V$ is the multiplicity of the conjugate irrep $T^{(\bar{\alpha})}$ in $V$. \label{dimmult}
\end{lemma}
\begin{proof}
Let us decompose the representations $U_g$ and $V_g$ into irreps:
\[
    U_g = \bigoplus_{\alpha\in \hat{\cg}} n_\alpha^U T_g^{(\alpha)}, \qquad
    V_g = \bigoplus_{\beta\in \hat{\cg}} n_\beta^V T_g^{(\beta)}.
\]
Then, we can write 
\begin{align*}
    U_g \otimes V_g &= \left( \bigoplus_{\alpha\in \hat{\cg}} n_\alpha^U T_g^{(\alpha)} \right) \otimes \left( \bigoplus_{\beta\in \hat{\cg}} n_\beta^V T_g^{(\beta)} \right) = \bigoplus_{\alpha, \beta\in \hat{\cg}} n_\alpha^U n_\beta^V (T_g^{(\alpha)} \otimes T_g^{(\beta)}) = \bigoplus_{\alpha, \beta, \gamma\in \hat{\cg}} n_\alpha^U n_\beta^V c_\gamma^{\alpha \beta} T_g^{(\gamma)},
\end{align*}
where the $c_\gamma^{\alpha \beta}$ are the Clebsch-Gordan coefficients.
Note that 
\begin{align*}
    c_1^{\alpha \beta}=\bracket{\chi^{(1)}}{\chi^{(\alpha)} \chi^{(\beta)}} = \int_{\cg} dg \chi^{(\alpha)}(g) \chi^{(\beta)}(g) = \bracket{\overline{\chi^{(\alpha)}}}{\chi^{(\beta)}} = \delta_{\bar{\alpha}\beta} .
\end{align*}
Thus, the dimension of the physical Hilbert space $\ch_\phys$ can be written as
\begin{align*}
    d_\phys=\dim\ch_\phys=\sum_{\alpha, \beta\in \hat{\cg}} n_\alpha^U n_\beta^V c_1^{\alpha \beta} = \sum_{\alpha, \beta\in \hat{\cg}} n_\alpha^U n_\beta^V \delta_{\bar{\alpha}\beta} = \sum_{\alpha, \beta\in \hat{\cg}} n_\alpha^U n_{\bar{\alpha}}^V .
\end{align*}
\vskip -2em
\end{proof}
The following result is typically only given for finite groups~\cite{Simon}; hence, for completeness, we here prove it for compact Lie groups. However, the result is certainly already well-known.
\begin{lemma}
The convolution of two irreducible characters yields
\begin{equation*}
    \chi^{\alpha}\ast\chi^{\beta}=\frac{\delta_{\alpha \beta}\chi^{\alpha}}{d_\alpha}.
\end{equation*} 
\end{lemma}
\begin{proof}
We use that the matrix elements of complex irreps $T^{(\alpha)}_{ij}$ of a compact Lie group $\mathcal{G}$ are orthogonal, i.e.
\begin{equation*}
    \bracket{T^{(\alpha)}_{ij}}{T^{(\beta)}_{kl}}=\int_{\mathcal{G}} dg \overline{(T^{(\alpha)}_g)_{ij}}(T_g^{(\beta)})_{kl}=\frac{\delta_{\alpha \beta}\delta_{ik}\delta_{jl}}{d_{\alpha}}.
\end{equation*}
So, let us check
\begin{eqnarray*}
\left(\chi^{(\alpha)}\ast\chi^{(\beta)}\right)(g)&=&\int_{\mathcal{G}}dh \chi^{(\alpha)}(gh^{-1})\ast\chi^{(\beta)}(h)=\sum_{ijk}\int_{\mathcal{G}}dh (T_g^{(\alpha)})_{ij}(T_{h^{-1}}^{(\alpha)})_{ji}(T_h^{(\beta)})_{kk}\nonumber\\
&=& \sum_{ijk}(T_g^{(\alpha)})_{ij}\int_{\mathcal{G}}dh\overline{(T_h^{(\alpha)})_{ij}}(T_h^{(\beta)})_{kk}=\sum_{ijk}(T_g^{(\alpha)})_{ij}\frac{\delta_{\alpha \beta}\delta_{ik}\delta_{jk}}{d_\alpha}=\frac{\delta_{\alpha \beta}\chi^{\alpha}}{d_\alpha}.
\end{eqnarray*}
\vskip -2em
\end{proof}

\section{Proof of Theorem \ref{TheResource}} \label{app:proof2}
\begin{proof}
Since we are interested in pure-state convertibility under covariant operations, we use the techniques of~\cite{Marvian,MarvianSpekkens} via \emph{characteristic functions} $\chi_{\varphi}(g):=\langle\varphi|V_g|\varphi\rangle$. Let us write $|\psi\rangle\stackrel{\mathcal{G}-{\rm cov}}\longrightarrow |\psi'\rangle$ if there is a $\mathcal{G}$-covariant operation $\mathcal{E}$ with $\mathcal{E}(|\psi\rangle\langle\psi|)=|\psi'\rangle\langle\psi'|$. Suppose that $|\psi\rangle_{S|R}\stackrel{\mathcal{G}-{\rm cov}}\longrightarrow |\psi'\rangle_{S|R}$, where $|\psi\rangle_{S|R}:=|\psi\rangle_{S|R}^{|e\rangle}$ and $|\psi'\rangle_{S|R}:=|\psi'\rangle_{S|R}^{|e\rangle'}$. Then there is a $\mathcal{G}$-covariant operation $\mathcal{E}$ on $S$ with $\mathcal{E}(\rho)=\rho'$, where $\rho=|\psi\rangle\langle\psi|_{S|R}$ and $\rho':=|\psi'\rangle\langle\psi'|_{S|R}$. Hence
\begin{eqnarray*}
\int_{\mathcal{G}}dg\,\mathcal{F}(\rho,V_g\rho V_g^\dagger)=\mathcal{U}(\psi_{S|R})=\mathcal{U}(\psi'_{S|R})
= \int_{\mathcal{G}}dg\, \mathcal{F}(\rho',V_g \rho' V_g^\dagger)=\int_{\mathcal{G}}dg\, \mathcal{F}(\mathcal{E}(\rho),\mathcal{E}(V_g\rho V_g^\dagger)).
\end{eqnarray*}
Since the fidelity satisfies the data-processing inequality, $\mathcal{F}(\mathcal{E}(\rho),\mathcal{E}(\sigma))\geq\mathcal{F}(\rho,\sigma)$, and $g\mapsto \mathcal{F}(\rho,V_g\rho V_g^\dagger)$ is continuous, we must have $\mathcal{F}(\rho,V_g\rho V_g^\dagger)=\mathcal{F}(\rho',V_g\rho'V_g^\dagger)$ for all $g\in\mathcal{G}$. But $\mathcal{F}(\rho,V_g\rho V_g^\dagger)=|\langle\psi|U_g|\psi\rangle_{S|R}|^2=|\chi_{\psi}(g)|^2$, and so $|\chi_{\psi}(g)|^2=|\chi_{\psi'}(g)|^2$ for all $g\in\mathcal{G}$. Furthermore, it follows from $|\psi\rangle_{S|R}\stackrel{\mathcal{G}-{\rm cov}}\longrightarrow |\psi'\rangle_{S|R}$ that there is a positive definite function~\cite[Theorem 63]{Marvian} $f:\mathcal{G}\to\mathbb{C}$ (see Supplemental Material~\ref{SecSomeLemmas} for a definition) such that $\chi_{\psi}(g)=\chi_{\psi'}(g)f(g)$ for all $g\in\mathcal{G}$. Now, for those $g$ with $\chi_{\psi}(g)\neq 0$, it follows that $|f(g)|=1$, i.e.\ $f(g)^{-1}=\overline{f(g)}$. Since $\chi_{\psi}(g)=0\Leftrightarrow \chi_{\psi'}(g)=0$,
\[
   \chi_{\psi'}(g)=\chi_{\psi}(g) \overline{f(g)}\mbox{ for all }g\in\mathcal{G}.
\]
If $f$ is positive definite, so is $\bar f$. Hence, it follows again from~\cite[Theorem 63]{Marvian} that $|\psi'\rangle_{S|R}\stackrel{\mathcal{G}-{\rm cov}}\longrightarrow |\psi\rangle_{S|R}$, which implies that $|\psi\rangle_{S|R}^{|e\rangle}$ is as asymmetric as $|\psi\rangle_{S|R}^{|e\rangle'}$.
\end{proof}

\section{Proof of Theorem \ref{TheAvg}} \label{app:proof3}
\begin{proof}
Using Lemma~\ref{trace of squared operator}, we have
\begin{align}
    \cu_{\rm phys} &= \int_{\ch_\phys} d\psi\ \tr[ \left( \Tr_R \ketbra{\psi}{\psi}_{RS} \right)^2] 
    = \tr\left(\mathbf{1}_{RR'} \otimes \mathbb{F}_{SS'} \int_{\ch_\phys} d\psi (\ketbra{\psi}{\psi}_{RS} \otimes \ketbra{\psi}{\psi}_{R'S'})\right) \nonumber\\ &= \frac{2}{d_{\rm phys} (d_{\rm phys} +1)} \tr \left((\mathbf{1}_{RR'} \otimes \mathbb{F}_{SS'}) \Pi_{\phys,\sym}^{RS,R'S'} \right)= \frac{2}{d_{\rm phys} (d_{\rm phys} +1)} \tr \left((\mathbf{1}_{RR'} \otimes \mathbb{F}_{SS'}) \Pi_{\sym}^{RS,R'S'} (\Pi_\phys^{RS}\otimes \Pi_\phys^{R'S'}) \right).
    \label{eq: integral in terms of projectors}
\end{align}
Let us first consider the expression
\begin{align*}
    (\mathbf{1}_{RR'} \otimes \mathbb{F}_{SS'}) \Pi_{\sym}^{RS,R'S'} = \frac{1}{2} (\mathbf{1}_{RR'} \otimes \mathbb{F}_{SS'}) (\mathbf{1} + \mathbb{F}_{RS,R'S'}) .
\end{align*}
Letting the second term act on a general basis state, we see
\begin{align*}
    (\mathbf{1}_{RR'} \otimes \mathbb{F}_{SS'})\mathbb{F}_{RS,R'S'} \ket{i}_R \ket{j}_{R'}\ket{s}_S \ket{k}_{S'} &= \mathbf{1}_{RR'} \otimes \mathbb{F}_{SS'} \ket{j}_R \ket{i}_{R'}\ket{k}_S \ket{s}_{S'} = \ket{j}_R \ket{i}_{R'}\ket{s}_S \ket{k}_{S'} \\
    &= \mathbb{F}_{RR'} \otimes \mathbf{1}_{SS'} \ket{i}_R \ket{j}_{R'}\ket{s}_S \ket{k}_{S'}.
\end{align*}
Thus,
\begin{align*}
    (\mathbf{1}_{RR'} \otimes \mathbb{F}_{SS'}) \Pi_{\sym}^{RS,R'S'} = \frac{1}{2} (\mathbf{1}_{RR'} \otimes \mathbb{F}_{SS'}+ \mathbb{F}_{RR'} \otimes \mathbf{1}_{SS'}).
\end{align*}
Hence, Eq.~\eqref{eq: integral in terms of projectors} can be written as
\begin{align*}
   \cu_{\rm phys} = \frac{1}{d_{\rm phys} (d_{\rm phys} +1)}\left[ \tr \left((\mathbf{1}_{RR'} \otimes \mathbb{F}_{SS'}) (\Pi_\phys^{RS}\otimes \Pi_\phys^{R'S'}) \right)+ \tr \left((\mathbb{F}_{RR'} \otimes \mathbf{1}_{SS'}) (\Pi_\phys^{RS}\otimes \Pi_\phys^{R'S'}) \right)\right].
\end{align*}
Let us rewrite the first term in the following way:
\begin{align*}
    \tr ((\mathbf{1}_{RR'} \otimes \mathbb{F}_{SS'}) (\Pi_\phys^{RS}\otimes \Pi_\phys^{R'S'}) ) &= \int_{\cg} dg \int_\cg dg' \tr ((\mathbf{1}_{RR'} \otimes \mathbb{F}_{SS'}) U_g^R \otimes V_g^S \otimes U_{g'}^{R'} \otimes V^{S'}_{g'} \nonumber\\ 
    &=\int_{\cg} dg \int_\cg dg' \chi^U(g) \chi^U(g') \tr(\mathbb{F}_{SS'}V_g^S \otimes V^{S'}_{g'}) \nonumber\\
    &=\int_{\cg} dg \int_\cg dg' \chi^U(g) \chi^U(g') \tr(V_g^S \cdot V^{S'}_{g'}) \nonumber\\
    &= \int_{\cg} dg \int_\cg dg' \chi^U(g) \chi^U(g')\chi^V(gg')
\end{align*}
where $\chi^U(g)=\tr(U_g) =  \sum_{\alpha \in \hat{\cg}} n_{\alpha}^U \tr(T_g^{(\alpha)}) =\sum_{\alpha \in \hat{\cg}} n_{\alpha}^U \chi^{(\alpha)}(g)$. At the first equality sign, we used $\Pi_\phys^{RS}=\int_{\cg} dg U_g^R\otimes V_g^S$ and similarly for $\Pi_\phys^{R'S'}$. We used Lemma \ref{tensor product flip trick} to go from the second to the third line. 

We can proceed similarly with the second term:
\begin{align*}
    \tr ((\mathbb{F}_{RR'} \otimes \mathbf{1}_{SS'}) (\Pi_\phys^{RS}\otimes \Pi_\phys^{R'S'}) ) = \int_{\cg} dg \int_\cg dg' \chi^V(g) \chi^V(g')\chi^U(gg').
\end{align*}

Note that in general, we can write
\begin{align*}
    \int_{\cg} dg \int_\cg dg' \chi^A(g) \chi^A(g')\chi^B(gg') = \int_{\cg} dg \int_\cg dh \chi^A(g) \chi^A(g^{-1}h)\chi^B(h) .
\end{align*}
Finally, we get
\begin{align}
    \cu_{\rm phys} &= \frac{1}{d_{\rm phys} (d_{\rm phys} +1)} \left( \int_{\cg} dg \int_\cg dh \chi^U(g) \chi^V(h)\chi^U(g^{-1}h) +  \int_{\cg} dg \int_\cg dh \chi^V(g) \chi^U(h)\chi^V(g^{-1}h)  \right) \nonumber\\
    &= \frac{1}{d_{\rm phys} (d_{\rm phys} +1)}\int_{\cg} dg \int_\cg dh \chi^U(g) \chi^V(h) \left( \chi^U(g^{-1}h) + \chi^V(h^{-1}g) \right). \label{eq: integrals as sum of characters}
\end{align}
Note that the first term can be written as
\begin{align*}
    \int_{\cg} dg \int_\cg dh \chi^U(g) \chi^V(h) \chi^U(g^{-1}h) &= \sum_{\alpha, \beta, \gamma \in \hat{\cg}} n_\alpha^U n_\beta^V n_\gamma^U \int_{\cg} dg \int_\cg dh \chi^{(\alpha)}(g) \chi^{(\beta)}(h) \chi^{(\gamma)}(g^{-1}h) \nonumber\\
    &= \sum_{\alpha, \beta, \gamma \in \hat{\cg}} n_\alpha^U n_\beta^V n_\gamma^U \int_\cg dh \chi^{(\beta)}(h) (\chi^{(\gamma)} * \chi^{(\alpha)})(h) \nonumber\\ 
    &= \sum_{\alpha, \beta \in \hat{\cg}} \frac{1}{d_\alpha} (n_\alpha^U)^2 n_\beta^V \bracket{\overline{\chi^{(\beta)}}}{\chi^{(\alpha)}} ,
\end{align*}
while the second term simplifies to 
\begin{align*}
    \int_{\cg} dg \int_\cg dh \chi^U(g) \chi^V(h) \chi^V(h^{-1}g) &= \sum_{\alpha,\beta \in \hat{\cg}} \frac{1}{d_\alpha} (n_\alpha^V)^2 n_\beta^U \bracket{\overline{\chi^{(\beta)}}}{\chi^{(\alpha)}} = \sum_{\alpha,\beta \in \hat{\cg}} \frac{1}{d_\beta} (n_\beta^V)^2 n_\alpha^U \bracket{\overline{\chi^{(\alpha)}}}{\chi^{(\beta)}} .
\end{align*}
Then, Eq.~\eqref{eq: integrals as sum of characters} gives
\begin{align}
     \cu_{\rm phys} &= \frac{1}{d_{\rm phys} (d_{\rm phys} +1)} \sum_{\alpha,\beta \in \hat{\cg}} \left( \frac{1}{d_\alpha} (n_\alpha^U)^2 n_\beta^V \bracket{\overline{\chi^{(\beta)}}}{\chi^{(\alpha)}} + \frac{1}{d_\beta} (n_\beta^V)^2 n_\alpha^U \bracket{\overline{\chi^{(\alpha)}}}{\chi^{(\beta)}} \right) \nonumber\\ 
     &= \frac{1}{d_{\rm phys} (d_{\rm phys} +1)} \sum_{\alpha,\beta \in \hat{\cg}} \frac{1}{d_\alpha} \left(  (n_\alpha^U)^2 n_\beta^V \bracket{\overline{\chi^{(\beta)}}}{\chi^{(\alpha)}} + (n_\alpha^V)^2 n_\beta^U \bracket{\overline{\chi^{(\beta)}}}{\chi^{(\alpha)}} \right) \nonumber\\
     &= \frac{1}{d_{\rm phys} (d_{\rm phys} +1)} \sum_{\alpha,\beta \in \hat{\cg}} \frac{1}{d_\alpha} \bracket{\overline{\chi^{(\beta)}}}{\chi^{(\alpha)}} \left( (n_\alpha^U)^2 n_\beta^V  + (n_\alpha^V)^2 n_\beta^U \right). \label{eq: integral as sum of char2} 
\end{align}
Since we are only considering the case $d_{\rm phys} \geq 1$, there exists at least one $\gamma$ such that $n_{\gamma}^U \neq 0$ and $n_{\bar{\gamma}}^V \neq 0$. Then, using Lemma \ref{dimmult}, we have
\begin{align*}
     \cu_{\rm phys} &=\frac{1}{d_{\rm phys} (d_{\rm phys} +1)} \sum_{\alpha \in \hat{\cg}} \frac{1}{d_\alpha} \left( (n_\alpha^U)^2 n_{\bar{\alpha}}^V  + (n_\alpha^V)^2 n_{\bar{\alpha}}^U \right) =\frac{1}{d_{\rm phys} (d_{\rm phys} +1)} \sum_{\alpha \in \hat{\cg}} \frac{n_\alpha^U n_{\bar{\alpha}}^V (n_{\alpha}^U + n_{\bar{\alpha}}^V)}{d_\alpha}   .
\end{align*}
\vskip -2em
\end{proof}

\section{(In)dependence of asymmetry on the coherent state system}\label{app:altmeasures}
We have seen that the conditional uniformity $\mathcal{U}(\psi_{S|R})=\int_{\mathcal{G}} dg\, |\langle\psi|_{S|R} V_g |\psi\rangle_{S|R} |^2$, and thus conditional asymmetry $\mathcal{A}=-\log\mathcal{U}$, does not depend on the choice of coherent state system $\{ |g\rangle_R \}_{g\in \mathcal{G}}$, even though $|\psi\rangle_{S|R}=|\psi\rangle_{S|R}^{|e\rangle}$ does. However, $\mathcal{A}$ provides just one possible \emph{measure of asymmetry}, or \emph{asymmetry monotone}, i.e.\ map $\mathcal{M}$ from quantum states to real numbers such that
\[
    \rho\mbox{ is at least as asymmetric as }\rho' \Rightarrow \mathcal{M}(\rho)\geq \mathcal{M}(\rho')
\]
(see Lemma~\ref{LemMon} below.) Here, we give an example of different asymmetry measures that \emph{do} depend on the choice of $|e\rangle$. But we will see that the resulting \emph{order}, i.e.\ which of the two given coherent state systems induces more asymmetry, is reversed by changing the asymmetry measure. This confirms the result of Theorem~\ref{TheResource}: no choice of coherent state system induces more asymmetry \emph{in the resource-theoretic sense} than any other.

To this end, consider the smallest finite non-Abelian group $\mathcal{G}=S_3$ acting on $R\otimes S$ with $R=S = \mathbb{C}^3$. We consider the case where $S_3$ is represented as permutations of the three basis vectors $\{ |i\rangle \}_{i=0}^2$, i.e. both $R$ and $S$ carry the fundamental representation. We write $U_g$ for the representation of the element $g\in S_3$ on $R$ and similarly $V_g$ on $S$. One can easily show that $U_g$ (and similarly $V_g=U_g$) takes the form 
\[
   U_g = T_g^{(1)} \oplus T_g^{(\mathrm{std})},
\]
where $T_g^{(1)}$ denotes the trivial and $T_g^{(\mathrm{std})}$ the two-dimensional standard representation of $S_3$.
The coherent state system $\{ |g\rangle_R \}_{g\in S_3}$ is generated by the seed state 
\[
   |e\rangle = \alpha |0\rangle + \beta |1\rangle + \gamma |2\rangle,
\]
where we have $U_g|e\rangle = |g\rangle,\ g\in S_3$.
There are two conditions on the coefficients $\alpha, \beta, \gamma \in \mathbb{C}$. First, for the states to be normalized, we require $|\alpha|^2 + |\beta|^2 +|\gamma|^2 = 1$. Moreover, the coherent state system generated by $|e\rangle$ needs to give rise to a resolution of the identity, i.e. $\frac{1}{|G|}\sum_{g \in S_3} |g\rangle\langle g | =c\cdot\mathbf{1}_R$ with $c\in\mathbb{R}$. With $|S_3|=6$, we find that $c=\frac{1}{3}$ and $| \langle + | e \rangle |^2 = \frac{1}{3}$ with $|+\rangle = \frac{1}{\sqrt{3}} ( |0\rangle +  |1\rangle +  |2\rangle )$. The latter condition is equivalent to 
\[
    \alpha \beta^* + \alpha^*\beta + \alpha \gamma^* + \alpha^*\gamma + \beta \gamma^* + \beta^* \gamma = 0.
\]
Let us choose two different seed states 
\begin{align*}
    |e\rangle_1 &= \frac{1}{3} |0\rangle - \frac{2}{3} |1\rangle - \frac{2}{3} |2\rangle, \\
    |e\rangle_2 &= \frac{1}{\sqrt{2}} |0\rangle  - \frac{i}{\sqrt{2}} |2\rangle,
\end{align*}
that generate two different systems of coherent states. One can easily check that the above conditions are satisfied. As expected, we find that the physical uniformity as defined in the main body takes on the same value for both systems. More precisely, by direct calculation, we find
\[
    \mathcal{U}_{2,{\rm phys}}^{|e\rangle_j}:= \int_{\mathcal{H}_{\rm phys}} d\psi\,  \frac{1}{6} \sum_{g\in S_3} \left|\langle\psi|_{S|R}^{|e\rangle_j} V_g |\psi\rangle_{S|R}^{|e\rangle_j} \right|^2 = \frac{1}{2} = \mathcal{U}_{\rm phys}
\]
for $j=1,2$, which confirms the result we get from Theorem \ref{TheAvg}: with $n^U_1 = n_1^V = 1$ and $n^U_{\mathrm{std}} = n_{\mathrm{std}}^V = 1$, we find $\mathcal{U}_{\phys}=\frac{1}{2}$. The subscript in the above equation indicates that we take the second power in the definition of the conditional uniformity. 

More generally, one can define
\[
    \mathcal{U}_{p,{\rm phys}}^{|e\rangle_j}:= \int_{\mathcal{H}_{\rm phys}} d\psi\,\mathcal{U}_{p}(\psi_{S|R}^{|e\rangle_j}),
\]
where $p\geq 0$ and
\[
    \mathcal{U}_p(\varphi_S)=\frac{1}{6} \sum_{g\in S_3} |\langle\varphi| V_g |\varphi\rangle_S |^p.
\]
Fig.~\ref{fig: asymmetry} illustrates $\mathcal{U}_{p,{\rm phys}}^{|e\rangle_j}$ as a function of $p$ for the two seed states $|e\rangle_1$, $|e\rangle_2$ above ($j=1,2$).
\begin{figure}[ht]
    \centering
    \includegraphics[width=0.7\linewidth]{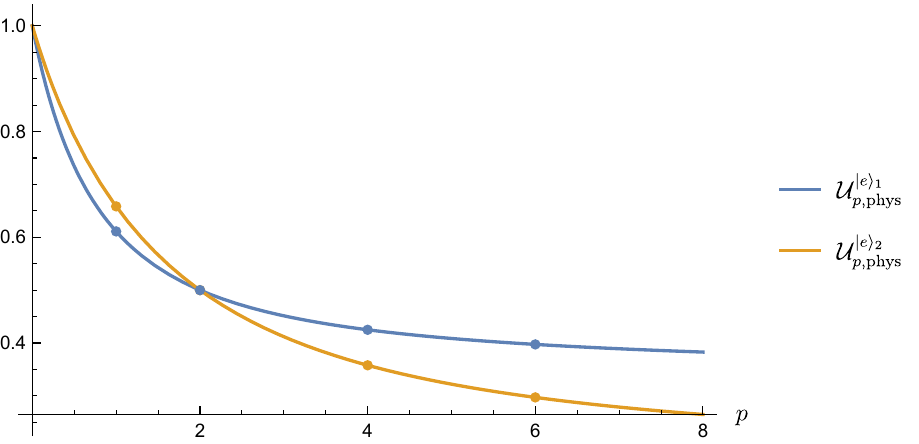}
    \caption{$\mathcal{U}_{p,{\rm phys}}^{|e\rangle_j}$ as a function of $p$ for $j=1,2$.}
    \label{fig: asymmetry}
\end{figure}

In particular, numerical and symbolic integration gives us the values
\[
   \mathcal{U}_{1,{\rm phys}}^{|e\rangle_1}\approx 0.611,\quad \mathcal{U}_{1,{\rm phys}}^{|e\rangle_2}\approx 0.658,\quad \mathcal{U}_{4,{\rm phys}}^{|e\rangle_1}=\frac{17}{40} \approx 0.425,\quad  \mathcal{U}_{4,{\rm phys}}^{|e\rangle_2} = \frac{229}{640} \approx 0.358.
\]
As shown in Lemma~\ref{LemMon} below, $\mathcal{A}_p:=-\log \mathcal{U}_p$ defines an asymmetry monotone. The above shows that the value of these monotones \emph{does} in general depend on the choice of coherent state system. For example, there exist states $|\psi\rangle_{RS}$ for which $\mathcal{A}_4(\psi_{S|R}^{|e\rangle_1})<\mathcal{A}_4(\psi_{S|R}^{|e\rangle_2})$.

Furthermore, neither one of the two coherent state systems can be regarded as ``better'' on average than the other one: we have $\mathcal{A}_{1,{\rm phys}}^{|e\rangle_1}:=-\log \mathcal{U}_{1,{\rm phys}}^{|e\rangle_1}>\mathcal{A}_{1,{\rm phys}}^{|e\rangle_2}$, but $\mathcal{A}_{4,{\rm phys}}^{|e\rangle_1}<\mathcal{A}_{4,{\rm phys}}^{|e\rangle_2}$. That is, the question of whether $|e\rangle_1$ or $|e\rangle_2$ induces more asymmetry on $S$ depends on the choice of monotone that is used to quantify the asymmetry.

This is not only true on average, but also on the level of individual states, which can be seen as follows. As apparent from the plot, we have $\displaystyle \left.\frac d {dp} \cu_{p,{\rm phys}}^{|e\rangle_1}\right|_{p=2}>\left.\frac d {dp} \cu_{p,{\rm phys}}^{|e\rangle_2}\right|_{p=2}$. Since $\displaystyle \frac d {dp}\cu_{p,{\rm phys}}^{|e\rangle_j}=\int_{\ch_{\rm phys}} d\psi\, \frac d {dp} \cu_p(\psi_{S|R}^{|e\rangle_j})$, there exist (many) physical states $|\psi\rangle_{RS}$ with $\displaystyle \left.\frac d {dp}\cu_p(\psi_{S|R}^{|e\rangle_1})\right|_{p=2}>\left.\frac d {dp}\cu_p(\psi_{S|R}^{|e\rangle_2})\right|_{p=2}$. But since $\cu_2(\psi_{S|R}^{|e\rangle_1})=\cu_2(\psi_{S|R}^{|e\rangle_2})$, the graphs of $\mathcal{U}_p(\psi_{S|R}^{|e\rangle_i})$ must ``cross'' at $p=2$, and hence there is some $\delta>0$ with
\[
   \mathcal{A}_{2-\delta}(\psi_{S|R}^{|e\rangle_2})<\mathcal{A}_{2-\delta}(\psi_{S|R}^{|e\rangle_1}),\mbox{ but }\mathcal{A}_{2+\delta}(\psi_{S|R}^{|e\rangle_2})>\mathcal{A}_{2+\delta}(\psi_{S|R}^{|e\rangle_1}).
\]
If we measure asymmetry via $\mathcal{A}_{2-\delta}$, then the seed coherent state $|e\rangle_1$ induces more asymmetry for such states $|\psi\rangle_{RS}$; but if we quantify asymmetry via $\mathcal{A}_{2+\delta}$, then we obtain the exact opposite. This illustrates the validity of Theorem~\ref{TheResource}: resource-theoretically, no choice of coherent state system induces more asymmetry than any other. In this specific example, this is true even \emph{after} averaging over the physical Hilbert space.
 \begin{lemma}
 \label{LemMon}
 For every compact Lie group $\mathcal{G}$, and for every $p\geq 0$, the quantity
 \[
    \mathcal{A}_p(\rho):=-\log\int_{\mathcal{G}} dg\, \mathcal{F}(\rho, V_g \rho V_g^\dagger)^{p/2}
 \]
 is an asymmetry monotone. This includes the case $\mathcal{A}=\mathcal{A}_2$ from the main text.
 \end{lemma}
 \begin{proof}
 Suppose that $\mathcal{E}$ is any $\mathcal{G}$-covariant map. Then, using that $x\mapsto x^{p/2}$ is non-decreasing and the data processing inequality $\mathcal{F}(\mathcal{E}(\rho),\mathcal{E}(\sigma))\geq\mathcal{F}(\rho,\sigma)$, we get
 \begin{eqnarray*}
 \mathcal{A}_p(\mathcal{E}(\rho))&=&-\log\int_{\mathcal{G}}dg\, \mathcal{F}(\mathcal{E}(\rho),V_g \mathcal{E}(\rho)V_g^\dagger)^{p/2}=-\log\int_{\mathcal{G}}dg\, \mathcal{F}(\mathcal{E}(\rho), \mathcal{E}(V_g\rho V_g^\dagger))^{p/2}\\
 &\leq& -\log\int_{\mathcal{G}}dg\, \mathcal{F}(\rho, V_g\rho V_g^\dagger)^{p/2}=\mathcal{A}_p(\rho).
\end{eqnarray*}
\vskip -2em
 \end{proof}

\end{document}